\documentclass[citeauthoryear]{llncs}

\begin{document}





\title{Monotonicity, Revenue Equivalence and Budgets}
\author{Ahuva Mu'alem\inst{1}}

\institute{Software Engineering Dept. Ort Braude College of Engineering, Karmiel, Israel. ahumu@yahoo.com
A significant part of this research was done while the author was 
a Post-Doctoral Fellow at the Social and Information Sciences Laboratory, Caltech.}

\maketitle
\begin{abstract}
We study multidimensional mechanism design in a common scenario 
where players have private information about 
their willingness to pay and their ability to pay.
We provide a complete characterization of dominant-strategy 
incentive-compatible direct mechanisms where over-reporting the budget is not possible. In several settings,  reporting larger budgets 
can be made suboptimal 
with a small randomized modification 
to the payments. 

We then derive a closely related partial characterization for the general case where players can arbitrarily misreport their private budgets. Immediate applications of these results include simple
 characterizations for mechanisms with publicly-known budgets and for mechanisms without monetary transfers.

The celebrated revenue equivalence theorem states that the seller's revenue for a broad class of standard auction formats and settings will
be the same in equilibrium. 
Our main application is a revenue equivalence theorem  
for financially constrained bidders.
\end{abstract}

\section{Introduction}

Budget constraints are central to big business 
auctions.~\cite{Cramton-95} states that it realistic to assume that all firms participating in the historic PCS telecommunications spectrum license auction, held on July 1994  by the Federal Communications Commission (FCC), were faced with budget constraints. In Google's  GSP  keyword auction and other search engine advertising  platforms, the bidders are required to specify their bids as well as 
their budget limits~(\cite{Ostrovsky},\cite{Varian} and~\cite{AMPP}).

Classical results in mechanism design literature do not  necessarily carry over to common scenarios 
with budget constraints.~\cite{Milgrom-02} showed 
that the dominant-strategy of the VCG auction breaks down when bidders have limited budgets.\footnote{See 
also~\cite{Nicole-Budget},~\cite{Lavi-May} and references therein.} 
~\cite{Che-Gale} showed a revenue dominance of the standard first-price auction over second-price auction in the presence of budget constraints.  On the positive side, several incentive-compatible auctions do carry over.
The classical multiple object ascending 
auction for unit-demand bidders 
due to~\cite*{DGS} generalizes so as 
to accommodate incentive-compatible auctions for financially constrained bidders~(\cite{Hatfield-Milgrom} and~\cite{AMPP}).
Recently,~\cite{Ron-Shahar-Noam}  showed that 
the clinching auction due to~\cite{Ausubel} generalizes  if bidders have publicly known budgets.~\cite{Conitzer-10} further showed that for
an infinitely divisible good, a bidder cannot improve her/his utility by reporting a budget smaller than the actual one. They 
also observed that  reporting larger budgets can be made suboptimal with a small randomized modification to the payments.

This paper studies dominant-strategy incentive-compatible 
deterministic mechanisms in a model with multidimensional types,
private values and private budgets with hard constraints (that is, no player can pay more than her actual budget).
We shall consider direct revelation mechanisms, which consist of a social choice function and a payment function. A social choice function aggregates  player reports of their private values as well as their  budgets 
to select one outcome. We say that a social 
choice function is  {\it implementable} if there exists a payment function that makes truthful revelation incentive-compatible. The property of a social 
choice function to be implementable by a unique payment function 
(up to an additive constant) is called {\it revenue equivalence}.

In classical mechanism design literature, implementability is intimately connected to monotonicity. 
In  single-item auctions, monotonicity requires that player who reports a higher value must have a greater probability 
of receiving the item.   
In general, monotonicity  is a necessary but not a sufficient condition for implementability of social choice functions. 
For single dimensional types 
(such as single-item auctions)~\cite{Myerson} 
showed that monotonicity is also sufficient for any social choice function to be implementable. However, for  multiple-item auctions (where players have a   
loosely correlated values for a distinct subset of items) this is no longer the case. For such multidimensional environments,~\cite{Rochet} showed that a condition called {\it cyclic-monotonicity} is necessary and sufficient for implementability. 
Combinatorially, this condition essentially says that a multidimensional social choice function is implementable if and only if a corresponding graph contains no negative cycle. 

Remarkably, monotonicity is also analogous to the celebrated revenue equivalence  principle for single-item auctions~(\cite{Myerson}). 
However, in general this need not be the case, since revenue equivalence in  multidimensional domains is not implied by cyclic-monotonicity. A recent result  by~\cite{Vohra-RevEquiv} characterizes revenue equivalence in  multidimensional domains.
The paper shows that revenue equivalence holds if and only if all distances in the above  corresponding graph are anti-symmetric. Intuitively, a strengthening of the cyclic-monotonicity condition characterizes revenue equivalence.

 While cyclic-monotonicity is a  more complex condition than monotonicity,
several studies have characterized multidimensional environments 
 for which monotonicity implies cyclic-monotonicity. 
 In a variety of multidimensional domains, monotonicity (or other related simple local properties)
 is essentially sufficient for 
 any social choice function to be 
 implementable (\cite{Moldevanu-2001},\cite{LMNB},\cite{LS}, 
 \cite{Monderer},\cite{Vohra-Book},\cite{Carroll}, 
 \cite{Archer-Kleinberg}, 
 and references therein).\footnote{\cite{Lavi-Swamy-cmon} used cyclic-monotonicity directly to design an incentive-compatible mechanism in a multidimensional scheduling setting.}

In this paper we find that Rochet's cyclic-monotonicity condition extends to  multidimensional environments with private budgets. 
We show how to circumvent the assumption that players always
have the ability to pay up to their respective values, 
while obtaining a characterization result and a corresponding graph in the same spirit as Rochet. Our second contribution is a generalization of the revenue equivalence characterization by~\cite{Vohra-RevEquiv} to  multidimensional environments with private budgets.


\subsection{Organization of this paper}
In Section~\ref{SEC-Mechanism Design without Money} we derive a simple characterization 
for multidimensional dominant-strategy incentive-compatible  mechanism design without money.  
Section~\ref{sect-private-Budgets} studies 
multidimensional dominant-strategy incentive-compatible mechanism design with private budgets. 
We start with a necessary condition
for incentive-compatible private budget settings. We next show the sufficiency of this 
condition  if over-reporting the budget is impossible. 
We then consider a strengthening of this condition and show its sufficiency for the general case (where players can report any budget).

Section~\ref{Sec:revenue-equiv} shows a revenue equivalence 
principle for private budget settings (where players can report any budget) satisfying certain reasonable conditions.
Section~\ref{Sec:example-and-application} illustrates the results through examples. All absent proofs can be found in the Appendix.

\section{Warm-up: Characterizing Mechanism Design without Money}\label{SEC-Mechanism Design without Money}
In this section we consider a multidimensional setting 
with private values where monetary transfers are infeasible. 
We show a simple characterization for players with cardinal preferences.

\subsection{The Model}
We consider a setting with $n$ players and a finite set 
 $\mathcal{A}$ of possible outcomes. 
Player $\ell$'s private valuation is determined by $v_\ell \in \mathcal{V}_\ell \subseteq \bbbr^{|\mathcal{A}|}$, 
 where $v_\ell(a)$ is player $\ell'$s value for outcome $a\in \mathcal{A}$, 
 and $\mathcal{V}_\ell$ is the space of all possible valuations of player $\ell$. 
Let  $\mathcal{V} = \mathcal{V}_1 \times \cdots \times \mathcal{V}_n$  denote the total valuation space. 
We consider a multidimensional  setting with private values 
(where $v_\ell(a)$ can be nonnegative or negative) and publicly known zero budgets. 

A {\it mechanism design without money} $f$ 
consists of a social choice function $f : \mathcal{V} \times \vec{0} \rightarrow \mathcal{A}$ and a zero payment scheme
$p : \mathcal{V} \rightarrow \vec{0}$,  where $\vec{0} \in \bbbr^n$. 
In a {\it direct revelation mechanism without money}, 
the social choice function chooses for a vector $v \in \mathcal{V}$ of aggregate reports of all players an outcome
 $f(v, \vec{0})$, whereas the payment scheme assigns a zero payment to each player. A mechanism design without money can be regarded 
as a mechanism for players with publicly-known zero budgets, where 
monetary transfers are not feasible. 
Let $(v_\ell, v_{-\ell})$ denote the aggregate report vector when player $\ell$ reports 
$v_\ell$ and the other players' reports are represented by $v_{-\ell}$.

\begin{definition} 
A direct revelation mechanism without money 
$f : \mathcal{V} \times \: \vec{0} \rightarrow \mathcal{A}$ is 
called {\rm dominant-strategy implementable} if 
for every player $\ell$, every $v_\ell, \; v'_\ell \in \mathcal{V}_\ell$  and  $v_{-\ell} \in  \mathcal{V}_{-\ell}$, 
the following condition holds:
\[
v_\ell(f(v, \; \vec{0}))  \ge v_\ell(f(v', \; \vec{0})),
\] 
where  $v = (v_\ell, \; v_{-\ell})$ and $\; v' = (v'_\ell,\;  v_{-\ell})$.
\end{definition}

\begin{remark}
{\rm When clear from context we will sometimes use the 
term {\it implementability}  to denote 
dominant-strategy implementability.} 
\end{remark}

\subsection{Presentation Assumptions}
\label{SUB-SEC:Presentation-Assumptions}
Throughout the  paper we without loss of generality 
restrict our attention to a model with a single agent (say, player $\ell$) and assume
reported valuations $v_{-\ell}$ and budgets of all other players $B_{-\ell}$ to be fixed. 
This is without loss of generality as all relevant definitions can be interpreted by holding all other players' reports fixed.
For simplicity of 
notation when clear from context we suppress the subscript $\ell$ and write $\mathcal{V}, v$ and $B$
 instead of $\mathcal{V}_\ell, v_\ell$, and $B_\ell$. 
For convenience, we also assume that $f$ is onto (since otherwise
$A$ can be condensed to be the range of $f$).

\subsection{The Characterization}

  For two outcomes $a$ and $a'$ let  

\begin{equation}
\delta(a, a') =  \inf \ \{v(a) - v(a') \ | \ v \in \mathcal{V} \mbox{ such that } 
f(v, \vec{0}) = a\}.
\label{eqn:delta-without-money}
\end{equation}

Informally, $\delta(a, a')$ represents the least gain achieved by deviation from  $a$ to $a'$, when $a$ is chosen by truthful reporting and $a'$ is chosen by misreporting (while keeping the valuation of other players fixed). Since $f$ is onto, $\delta(a, a') <  \infty$. 
We can now state our simple  characterization for 
implementable mechanisms without money.

\begin{proposition}\label{thm-Characterization-for-zero-budgets}
A mechanism without money  $f: \mathcal{V} \times \: \vec{0} \rightarrow \mathcal{A}$ is 
dominant-strategy implementable  if and only if
$\delta(a, a') \ge 0$ for every $a, \; a'\in \mathcal{A}$.
\end{proposition}

\begin{proof} 
Assume $f$ is implementable without money. 
For any pair of valuations $v, v' \in \mathcal{V}$ such that  $f(v, \vec{0}) = a$  and 
 $f(v', \vec{0}) = a'$,  we have  $v(a) \ge v(a')$ by the implementability without money of $f$.
The definition of $\delta$ then gives 
 $v(a) - v(a') \ge \delta(a, a') \ge 0$.
 
 Conversely, if $f$ is not implementable without money, there 
 exist $v, v' \in \mathcal{V}$ such that  $v(a) <  v(a')$,  where  $f(v, \vec{0}) = a$  and 
 $f(v', \vec{0}) = a'$.  Then, $\delta(a, a') \le  v(a) - v(a') <  0$.\qed
\end{proof}

\section{Characterizing Mechanism Design with Private Budgets}\label{sect-private-Budgets}

In this section we study dominant-strategy incentive-compatible 
mechanisms in a multidimensional setting where players 
have private valuations and private budgets. 

We continue with the private value setting introduced  previously
 using the  following modifications. 
Player $\ell$ has private valuation
$v_\ell \in \mathcal{V}_\ell \subseteq \bbbr_+^{|\mathcal{A}|}$, 
 where $v_\ell(a)$ is player $\ell'$s {\it nonnegative} value for outcome 
 $a\in \mathcal{A}$. In addition, player $\ell$ has private  budget $B_\ell \ge 0$. 
 Let $\mathcal{B}_\ell \subseteq \bbbr_+ \cup \infty$ denote
 the space of all possible budget limits of player $\ell$.
 Notice that if the budget of player $\ell$ is {\it public knowledge} then the set $\mathcal{B}_\ell$ is a singleton. 
 
 Players are assumed to be utility maximizers  but can never pay beyond their budgets. 
 Specifically, player  $\ell$'s 
 utility with  private budget $B_\ell$, derived for paying $x \ge 0$ for outcome $a$,  
 is $v_\ell(a) - x$ as long as $x \le B_\ell$ and is negative infinity if 
 $x > B_\ell$.  
Note that we consider here a multidimensional  setting with
 {\it nonnegative} private values and private budget constraints.

A mechanism $(f, \; p)$ with private budgets consists of a social-choice function
$f : \mathcal{V} \times \mathcal{B} \rightarrow \mathcal{A}$ and a payment function  
$p : \mathcal{V} \times \mathcal{B} \rightarrow \bbbr^n$, where $\mathcal{V} = \mathcal{V}_1 \times \cdots \times \mathcal{V}_n$, 
$\mathcal{B} = \mathcal{B}_1 \times \cdots \times \mathcal{B}_n$ and 
the $\ell^{\mathrm{th}}$  component of $p$ is the payment requirement of player $\ell$. 
We restrict attention to direct revelation mechanisms 
where each player 
reports its private valuation $v_\ell \in \mathcal{V}_\ell$, as well as 
its private budget $B_\ell \in \mathcal{B}_\ell$.    
Let $v_{-\ell}, B_{-\ell}$ denote the reported valuations and budgets of players 
other than  $\ell$.

\begin{definition} 
A direct revelation mechanism $(f,\; p)$ with  private budgets is 
called {\rm dominant-strategy implementable}  (or {\rm implementable}, for short) 
if the following conditions hold: 
\end{definition}

\begin{enumerate}

\item
{\bf (IC)} Incentive Compatibility:  
For  every player $\ell =1, \ldots, n$, every $v \times B \in \mathcal{V} \times \mathcal{B}$,  and 
every $v'_\ell \times B'_\ell \in \mathcal{V}_\ell \times \mathcal{B}_\ell$ if we denote 
$a=f(v, B)$ and $a'=f((v'_\ell, v_{-\ell}) , (B'_\ell , B_{-\ell}))$, then 

\[
v_\ell(a)  - \ p_\ell(v, B) \ge  
v_\ell(a') \ - \ p_\ell((v'_\ell, v_{-\ell}) , (B'_\ell , B_{-\ell})), 
\]
\[
 \mbox{for every } p_\ell((v'_\ell, v_{-\ell}) , (B'_\ell , B_{-\ell})) \le B_\ell. 
\]

i.e., no player can improve its utility by misreporting its true private valuation and/or true private budget in order to obtain some other outcome whose payment is below $B_\ell$.
If the budget is {\it public knowledge} then players can only report their true budgets. If {\it budget over-reporting} is not allowed, then 
$B'_\ell \le B_\ell$. 
\item 
 {\bf (IR)} Individual Rationality: 
The mechanism never requires a player to pay more than its reported value.
\item 
 {\bf (BF)} Budget Feasibility: The mechanism never requires a player to pay more than its reported budget.
\item
 {\bf (NPT)} No Positive Transfer:
The mechanism never subsidizes any player with a monetary amount. 
\end{enumerate}

\begin{definition} 
A social choice function $f$ 
with private budgets is called  
{\rm implementable}  if there exists a 
payment function $p$ such that the mechanism $(f, \;p)$ is dominant-strategy implementable. 
\end{definition}

\paragraph*{Graph Theoretic Definitions.} To state our main results we use some basic definitions 
from Graph Theory.  Recall that a (directed) {\it graph} is a pair $G(M, E)$, 
where $M$ is a finite set  and $E \subseteq M \times M$.  
An element $n \in M$ is called a node and an element $e=(n, n')\in E$ is called an arc.  A {\it complete} directed graph is a graph in which $E = M \times M$.

A (finite) {\it path} from node $n_1$ to node $n_k$ in $G$ is a  sequence $P=(n_1, n_2, \ldots, n_k)$ in $E$ such that  
$e_i =(n_i, m_i)  \in E$ and $n_{i+1} = m_i$, where  $i=1, \ldots, k-1$.    
We denote by $e\in P$ that $e$ is an arc in $P$. A {\it (finite) cycle} in $G$ is a path  $C=(e_1, e_2, \ldots, e_k)$ such that $m_k = n_1$. 
 When clear from the context,  the cycle $C$ may be written as $(n_1, n_2, \ldots, n_k)$.
 To each arc $e \in E$ we assign a {\it length} $l(e) \in \bbbr$. 
The length of a path $P$ is  $l(P)=\Sigma_{e \in P\;} l(e)$. 
A path $P$ is called negative if $l(P) < 0$. 
We denote by $\Delta (n_i, n_k)$ the length of the 
{\it shortest path} from node $n_i$ to node $n_k$ 
in $G$.

A {\it strongly connected component} of a directed graph $G=(M, E)$ 
is a maximal set of vertices $K \subseteq M$ such that for every pair of nodes $n_i$ and $n_j$ in $K$, there is a path 
from $n_i$ to $n_j$, and a path from $n_j$ to $n_i$.
 A {\it directed acyclic graph} is a directed graph with no cycles.

\subsection{Necessary and Sufficient Conditions for Private Budgets}
Our main result in this section is a necessary and sufficient conditions for the multidimensional private budget setting.
Theorem~\ref{thm:necessary-cond-private-budgets} provides a necessary condition for implementability. Theorem~\ref{thm:suff-private-budgets-no-over-report} and  Theorem~\ref{thm:sufficiency-if-G=widehat-G}
show the {\it sufficiency}  of the condition in Theorem \ref{thm:necessary-cond-private-budgets} in several settings.
The proof technique is constructive;  namely, a concrete payment function is specified to show the sufficiency. 
Examples~\ref{EXAMPLE:G-vs-widehat-G} and~\ref{EXAMPLE:non-necessity-of-widehat-G} in Section~\ref{Sec:example-and-application} 
further demonstrate our conditions.

For $a \in \mathcal{A}$, we define
\begin{equation}
\beta(a)  = \inf \ \{B \ | \ f(v, B) = a \mbox { for some } v \in \mathcal{V} \},  
\label{eqn:beta-private} 
\end{equation}

\begin{equation}
\omega(a)  = \inf \ \{v(a) \ |  \ f(v, B) = a \mbox { for some } B \in \mathcal{B} \}, 
\label{eqn:omega-private} 
\end{equation}

\begin{equation}
\theta(a)  = \min \ \{ \beta(a), \ \omega(a) \}. 
\label{eqn:theta-private}
\end{equation}

i.e.,  $\beta(a)$ is the minimum reported budget required 
to obtain outcome $a$, 
 and $\omega(a)$ is the minimum reported valuation required 
to obtain outcome $a$.
Therefore, $\theta(a)$ serves as an 
upper bound on any budget-feasible individually-rational payment for outcome $a$, if exists. Since $f$ is onto, $\theta(a) < \infty$. 

The next definition is a 
generalization of~(\ref{eqn:delta-without-money}) to the private budget setting. 
For $a, a' \in \mathcal{A}$, let 

\begin{equation}
\delta(a, a') =  \inf \ \{v(a) - v(a') \ | \  
f(v, B) = a, \mbox{ where } v \in \mathcal{V}  \mbox{ and }  B \in  \mathcal{B} \cap \left[ \beta(a'), \infty \right)   \}.
\label{eqn:delta-private}
\end{equation}

We follow the convention that the infimum over an empty set equals 
$\infty$ (observe that if $\beta(a') > \beta(a)$, it can be the case that $f(v, B) \neq a$ for every $v, B$ with $B \ge \beta(a')$). Furthermore, if $f$ is implementable then  $\delta(a, a') > -\infty$ 
(see Claim~\ref{claim-delta-neq-neg-infty}).

Using $\delta$ and $\theta$, we can now construct a   
finite directed graph associated 
with the social choice function $f$.  

\begin{definition}[The Graph $G_{f, \; \mathcal{B}}$]\label{def:G-f-graph}
Let $G_{f, \; \mathcal{B}}$ be a complete directed graph over the nodes  
${M}(G_{f, \;\mathcal{B}}) = \{1, 2, \ldots, |\mathcal{A}|\}$, 
i.e.,  the nodes of the graph correspond to all possible outcomes. 
The {\rm length} of a directed arc $(i, k) \in {E}(G_{f, \;\mathcal{B}})$ 
is defined as 
\[ 
l(i, \; k) = \min \{ \delta(a_i,\;  a_k), \; \theta(a_i)\}.  
\] 
\end{definition}

We can now state our necessity result: 

\begin{theorem}\label{thm:necessary-cond-private-budgets}
If a social choice function $f: \mathcal{V} \times \mathcal{B} \rightarrow \mathcal{A}$ 
with private budgets is  implementable  then 
 the corresponding graph $G_{f, \; \mathcal{B}}$ 
 contains no negative cycles.
\end{theorem}

We now consider the case where over-reporting the budget is 
impossible and show the {\it sufficiency} of the condition in Theorem
 \ref{thm:necessary-cond-private-budgets}. 
This allows us to  derive a complete characterization of the class
 of implementable social choice functions where over-reporting the private budget is impossible.  
This class is rather general, it includes mechanisms without money, mechanisms for players with publicly-known budgets and more. 
In several settings, reporting larger budgets can be made a dominated strategy (suboptimal) 
by using a small randomized modification 
to the payments~(\cite{Conitzer-10}).

\begin{theorem}\label{thm:suff-private-budgets-no-over-report}
A social choice function $f: \mathcal{V} \times \mathcal{B} \rightarrow \mathcal{A}$ 
for private budgets with no budget over-reporting is 
 implementable  if and only if 
 the corresponding graph $G_{f, \; \mathcal{B}}$ 
 contains no negative  cycles.
\end{theorem}

Example~\ref{EXAMPLE:G-vs-widehat-G} shows that the no budget over-reporting assumption is crucial to the sufficiency in Theorem~\ref{thm:suff-private-budgets-no-over-report}. Notice that Proposition~\ref{thm-Characterization-for-zero-budgets} is not a special case of 
 Theorem~\ref{thm:suff-private-budgets-no-over-report}, 
 since here we restrict the values to be nonnegative.

We now consider the general setting of private budgets 
(assuming players can report any budget limit) and provide a sufficient condition for implementability,
that differs from the necessary condition in Theorem~\ref{thm:necessary-cond-private-budgets}.  Our sufficient condition for implementability requires a subtle but crucial change in~(\ref{eqn:delta-private}). 
For $a, a' \in \mathcal{A}$, we define
 
\begin{equation}
\widehat{\delta}(a, a') =  \inf \ \{v(a) - v(a') \ | \  f(v, B) = a, \mbox{ where } v \in \mathcal{V} \mbox{ and  }  B \in \mathcal{B}\}.
\label{eqn:delta-private-suff}
\end{equation}

Clearly, $\widehat{\delta}(a, a') \le \delta(a, a')$. 
Additionally, if the budget is publicly known or if $\beta(a') \le \beta(a)$ then $\widehat{\delta}(a, a') = \delta(a, a')$.

\begin{definition}[The Graph $\widehat{G}_{f, \; \mathcal{B}}$]\label{def:widehat-G-f-graph}
Let $\widehat{G}_{f, \; \mathcal{B}}$ be a complete directed graph over the nodes  
${M}(G_{f, \;\mathcal{B}}) = \{1, 2, \ldots, |\mathcal{A}|\}$, 
i.e.,  the nodes of the graph correspond to all possible outcomes. 
The {\rm length} of a directed arc $(i, k) \in {E}(G_{f, \;\mathcal{B}})$ 
is defined as 
\[ 
l(i, \; k) = \min \{ \widehat{\delta}(a_i,\;  a_k), \; \theta(a_i)\}.  
\] 
\end{definition}

If the graph $\widehat{G}_{f, \; \mathcal{B}}$ 
contains no negative cycles then so does ${G}_{f, \; \mathcal{B}}$, but not vice versa. 

We can now state our sufficiency results for private budgets 
(assuming players can report any budget limit).

\begin{proposition}\label{prop:suff-private-budget} 
Let $f: \mathcal{V} \times \mathcal{B}\rightarrow \mathcal{A}$ 
be a social choice function with private budgets. 
If the corresponding graph $\widehat{G}_{f, \; \mathcal{B}}$ 
contains no negative  cycles
then $f: \mathcal{V} \times \mathcal{B}\rightarrow \mathcal{A}$ is implementable.
\end{proposition}

Example~\ref{EXAMPLE:non-necessity-of-widehat-G}  shows an implementable social choice function whose corresponding graph $\widehat{G}_{f, \; \mathcal{B}}$ contains a negative cycle.  However, if $\widehat{G}_{f, \; \mathcal{B}} = {G}_{f, \; \mathcal{B}}$, then 
the condition in Proposition~\ref{prop:suff-private-budget}
is  also necessary.

\begin{theorem}\label{thm:sufficiency-if-G=widehat-G}
Let $f: \mathcal{V} \times \mathcal{B}\rightarrow \mathcal{A}$ 
be a social choice function with private budgets.
If $\widehat{\delta}(a, a') = \delta(a, a')$ for every $a, a' \in \mathcal{A}$ 
then  $f: \mathcal{V} \times \mathcal{B} \rightarrow \mathcal{A}$   is 
 implementable  if and only if 
 the corresponding graph $G_{f, \; \mathcal{B}}$ 
 contains no negative  cycles.
\end{theorem}

\begin{proof}
Immediate from  Theorem~\ref{thm:necessary-cond-private-budgets} and Proposition~\ref{prop:suff-private-budget} 
(notice that Theorem~\ref{thm:suff-private-budgets-no-over-report} is of no use here,
since players can over-report their budget limits).\qed
\end{proof}

\section{Revenue Equivalence with Budgets}\label{Sec:revenue-equiv}

The celebrated revenue equivalence principle  says that
any two payment function implementing the same social choice function differ by a constant, and thus the 
payment function is uniquely defined up to an additive constant. 
The uniqueness for budget constraints players is slightly more subtle. 

Recall that by Lemma~\ref{lemma:taxation-Principle} (using the 
presentation assumptions in Subsection~\ref{SUB-SEC:Presentation-Assumptions}) 
if  $f(v, B)=f(v', B')=a$ for some $a \in \mathcal{A}$ we have that $p(v, B)=p(v, B')=p_{a}$. 
When players have budget constraints, 
any two payment functions implementing the same $f: \mathcal{V} \times \mathcal{B} \rightarrow \mathcal{A}$  might be differ by a {\it collection} of constants. Intuitively, there might be a distinct constant for every budget level $\beta(a)$, where $\beta(a)$ denotes the minimum reported budget 
required to obtain outcome $a \in \mathcal{A}$.

\begin{definition}[Revenue Equivalence]
An implementable social choice function 
$f: \mathcal{V} \times \mathcal{B}\rightarrow \mathcal{A}$ 
satisfies the revenue equivalence principle 
if for every two dominant strategy implementable mechanisms 
$(f, p)$ and $(f, p')$, 
we have that $\beta(a) = 
\beta(a')$ implies that 
\[
p_{a}-  p'_{a} =  p_{a'}-  p'_{a'}.
\]
\end{definition}

Notice that if  $\beta(a) \neq \beta(\widehat{a})$
 then $p_{\widehat{a}} -  
 p'_{\widehat{a}}$ need not be equal to $p_{a} -  p'_{a}$.

\subsection{Assumptions}
We shall show, under several reasonable assumptions, that a necessary and sufficient condition for revenue equivalence with private budget do exists.
We now state our assumptions on $f$. 

\begin{definition}[Generic Implementation at $a_i$]\label{def:Generically-implementable}
Let $f: \mathcal{V} \times \mathcal{B}\rightarrow \mathcal{A}$ be a social choice function, and let $a_i \in \mathcal{A}$. We say that 
a payment function  $p: \mathcal{V}\times \mathcal{B}\rightarrow \bbbr$ 
{\rm generically implements}  $f$ {\rm at} $a_i$  if 
\begin{enumerate}
\item 
$(f, p)$ is a dominant strategy implementable mechanism.

\item 
Let $a_j \in \mathcal{A}$. 
If $\beta(a_j) = \beta(a_i) > 0$ then 
$0 < p_{a_j} < \theta(a_j)$.

\item
Let $a_j, a_k \in \mathcal{A}$. 
If  $\beta(a_k) < \beta(a_j) = \beta(a_i)$
then $p_{a_j} - p_{a_k} \neq  \delta(a_j, a_k)$.

\item
Let $a_j, a_k \in \mathcal{A}$. 
If $\beta(a_k) > \beta(a_j) = \beta(a_i)$, $f(v, b) = a_j$ and 
$v(a_j)-v(a_k) \ge p_{a_j} - p_{a_k}$ then 
$v(a_j)-v(a_k) > p_{a_j} - p_{a_k}$. 
\end{enumerate}

If $p: \mathcal{V} \times \mathcal{B}\rightarrow \bbbr$ 
generically implements
$f$ at $a_i$, we say that $f$ is 
{\rm generically implementable at} $a_i$. 
\end{definition}

Intuitively, the assumption requires strict inequalities 
(recall Claim~\ref{claim-delta-neq-neg-infty}). 
Notice that  in Definition~\ref{def:Generically-implementable},  part~(4) it might be the case that $p_{a_k} > \beta(a_j)$.

\begin{definition}\label{def:Generically-implementable-at-all-outcomes}
A social choice function $f: \mathcal{V} \times \mathcal{B}\rightarrow \mathcal{A}$
is {\rm generically implementable} if 
for every $a_i \in \mathcal{A}$ there exists a payment function  $p^i: \mathcal{V} \times \mathcal{B}\rightarrow \bbbr$ that 
generically implements
$f$ at $a_i$.
\end{definition}

\subsection{Characterization}\label{SUB-SEC:REV-EQUIV-THM}
In this subsection we extend a recent characterization of revenue equivalence for multidimensional domains by~\cite{Vohra-RevEquiv}. 
We state our characterization result for private budgets and then  briefly discuss the differences between the proof techniques.  
Recall that $\Delta_{G_{f, \; \mathcal{B}}}(i, i')$ is the length of the 
{\it shortest path} from node $i$ to node $i'$ 
in the graph $G_{f, \; \mathcal{B}}$. 

\begin{theorem}[Characterization of Revenue Equivalence]\label{thm:rev-equiv}
A generically implementable social choice function 
$f: \mathcal{V} \times \mathcal{B}\rightarrow \mathcal{A}$ for private budgets  
satisfies the revenue equivalence principle 
if and only if in
all corresponding  graphs $G_{f, \; \mathcal{B}}$
obtained from a combination of a player and a reported valuations  and budgets of the
other players we have that $\Delta_{G_{f, \; \mathcal{B}}}(i, k)= -\Delta_{G_{f, \; \mathcal{B}}}(k, i)$
 for all $a_i, a_k \in \mathcal{A}$ with $\beta(a_i)=\beta(a_k)$.
\end{theorem}

Observe that $f: \mathcal{V} \times \mathcal{B}\rightarrow \mathcal{A}$ is implementable, and thus the proof of Theorem~\ref{thm:rev-equiv} only requires the necessary condition in Theorem~\ref{thm:necessary-cond-private-budgets}. 
Therefore, Theorem~\ref{thm:rev-equiv} is applicable for  private budget settings (where players can report any budget limit) satisfying the condition in Definition~\ref{def:Generically-implementable-at-all-outcomes}.

Essentially, the original proof of~\cite{Vohra-RevEquiv} constructs certain payments based on the specific  shortest paths of the graph. However, if $\Delta_{G_{f, \; \mathcal{B}}}(i, k)= -\Delta_{G_{f, \; \mathcal{B}}}(k, i)$ holds then this implies that these payments can be negative. 
In our setting, negative payments are excluded by the no positive transfer requirement. Instead, 
our necessity proof (Claim~\ref{claim-rev-equiv-necessary}) requires a subtle strongly-component argument and thus is different from theirs.  
Our sufficiency proof (Claim~\ref{claim-rev-equiv-sufficieny}) is based on a straight-forward adaptation of~\cite{Vohra-RevEquiv} and on 
Theorem~\ref{thm:necessary-cond-private-budgets}.

\section{Examples}\label{Sec:example-and-application}

This section illustrates our characterization results through some examples.

\begin{example}
\label{EXAMPLE:G-vs-widehat-G}
{\rm 
We show that the no budget over-reporting assumption is crucial to the sufficiency in Theorem~\ref{thm:suff-private-budgets-no-over-report}.
Consider a single player with possible outcomes $\mathcal{A}=\{a, a'\}$, possible private budgets $\mathcal{B}=\{10, 20\}$, and possible values $\mathcal{V}=\{v, v'\}$, where $v(a)=20, v(a')=10$ and $v'(a)=11, v'(a')=0$. 
Let $f(v, 20)=f(v',20)=a$ and $f(v, 10)=f(v',10)=a'$. 
If over reporting the budget is not allowed, then the payment function $p_{a}=10$ and $p_{a'}=0$ implements $f$.
However, if reporting any budget is allowed then no payment implements $f$. To see this, observe that $p_{a'}$ must be $0$ since  
$f(v',10)=a'$. Then, since $f(v, 20)=a$ we have $10 = v(a)-v(a') \ge p_a - p_{a'}$, and therefore $p_a \le 10$. But then if the true type of the player is $(v', 10)$ it is beneficial to misreport $(v', 20)$. Observe that $G_{f, \; \mathcal{B}}$ contains no negative cycle, since 
$\min\{\delta(a, a'), \; \theta(a, a')\}=10$,
$\min\{\delta(a', a), \; \theta(a', a)\}=\min\{\infty, 0\} = 0$.

In this example we also have that 
$\widehat{G}_{f, \; \mathcal{B}}$ contains a negative cycle (since $\min\{\widehat{\delta}(a, a'), \; \theta(a, a')\}=10$,
$\min\{\widehat{\delta}(a', a), \; \theta(a', a)\}=-11$) and thus 
$\widehat{G}_{f, \; \mathcal{B}} \neq G_{f, \; \mathcal{B}}$.
}
\end{example}

\begin{example}
\label{EXAMPLE:non-necessity-of-widehat-G}
{\rm 
We show an implementable $f$ whose corresponding graph $\widehat{G}_{f, \; \mathcal{B}}$ contains a negative cycle (and therefore 
the condition in Proposition~\ref{prop:suff-private-budget}
is  sufficient but not necessary for implementability). 
Consider a single player with possible outcomes $\mathcal{A}=\{a, a'\}$, possible private budgets $\mathcal{B}=\{1, 5\}$, and possible values $\mathcal{V}=\{v, v'\}$, where $v(a)=20, v(a')=10$ and $v'(a)=10, v'(a')=0$. 
Let $f(v, 5)=f(v',5)=a$ and $f(v, 1)=f(v',1)=a'$. 
Clearly, the payment function $p_{a}=5$ and $p_{a'}=0$ implements $f$. However, $\min\{\widehat{\delta}(a, a'), \; \theta(a, a')\}=5$,
$\min\{\widehat{\delta}(a', a), \; \theta(a', a)\}=-10$, that is $(a, a')$ is a negative cycle in $\widehat{G}_{f, \; \mathcal{B}}$.
To demonstrate Claim~\ref{claim-delta-neq-neg-infty} notice that 
$\widehat{\delta}(a', a) < p_{a'} - p_a = -5 \le \delta(a', a)=\infty$. 

}
\end{example}

We next consider  a broad class of social choice functions  
known as affine maximizers.
This class  encompasses the extensively studied class of weighted VCG mechanisms. 
We begin with the definition of  this class.

\begin{definition}
A social choice function $f$ is called an {\rm affine maximizer}
if for some player weights $\kappa_1 > 0, \ldots, \kappa_n > 0$
and some outcome weights $\gamma_{a} \in \bbbr$ for every ${a} \in \mathcal{A}$,
we have that:
\[
f(v_1, \ldots, v_n) \in \mbox{\rm argmax}_{ \ a \in \mathcal{A}\;} 
\{\sum_\ell \kappa_\ell \cdot v_\ell(a) + \gamma_{a}\}.
\]
\end{definition}

We next consider a natural  mechanism without money for 0/1 valuations (intuitively, each player can vote for all his most preferred alternatives, assuming all are equally desired).

\begin{claim}\label{claim-VCG-possibility-without-money}
The mechanism without money 
$f(v_1, \ldots, v_n) \in \mbox{\rm argmax}_{ \ a \in \mathcal{A}\;} \{\sum_\ell \kappa_\ell \cdot v_\ell(a) + \gamma_{a}\}$, where  
$v_\ell(a) \in \{0, 1\}$ for every $a \in \mathcal{A}$ and $\ell \in \{1, \ldots, n\}$ 
(assuming that ties among outcomes are broken lexicographically)
satisfies the condition in Proposition~\ref{thm-Characterization-for-zero-budgets} 
and therefore  is  dominant-strategy implementable. 
\end{claim}


As stated informally in~\cite{Nicole-Budget} no 
affine maximizer  is implementable with private budgets  if the players can specify  their budget limits 
 in addition to reporting their valuations (since players can 
over-report their values while under-reporting their budgets to avoid charges). 
We use  our characterization  to address the publicly-known budget case (where  
 players cannot misreport their budgets).

\begin{claim}\label{prop-VCG-impossibility-public}
Suppose that $\mathcal{V}_\ell = \bbbr_+^{|\mathcal{A}|}$ for every player $\ell$ and 
let $n \ge 2$.   
If at least one player has budget $ < \infty$, then no 
affine maximizer  is implementable with publicly-known  budgets.
\end{claim}

\section*{Acknowledgments}
I thank Sushil Bikhchandani, Shahar Dobzinski, Federico  Echenique,  David Kempe, Ron Lavi, John  Ledyard, Alexander Linden, Noam Nisan and Motty Perry  for early discussions and helpful suggestions.

\section{Appendix: Proofs}

\subsection{Proof of Theorem~\ref{thm:necessary-cond-private-budgets}}

The proof of Theorem~\ref{thm:necessary-cond-private-budgets} 
is the consequence of the following taxation principle lemma and claims.

\begin{lemma}[\cite{Nicole-Budget}]\label{lemma:taxation-Principle} 
Let  $(f, \; p)$ be an implementable mechanism for private budgets. 
Let 
\[
\mathcal{A}(v_{-\ell}, B_{-\ell}) = \{ a_k \in \mathcal{A} \ | 
\ \exists v'_\ell \in \mathcal{V}_l  \mbox { and } \exists B'_\ell \in \mathcal{B}_{\ell}  
\mbox{ such that } f((v'_\ell, v_{-\ell}), (B'_\ell, B_{-\ell})) = a_k \}.
\]
The following conditions hold: 
\end{lemma} 
\begin{enumerate}
\item
For every $\ell$, every $v_{-\ell}$ and every $B_{-\ell}$ 
there exist prices $p_{a_k} \in \bbbr$, 
for every $a_k \in \mathcal{A}(v_{-\ell}, B_{-\ell})$, such that for all 
$v_\ell$ and all $B_\ell$ with 
$f((v_\ell, v_{-\ell}), (B_\ell, B_{-\ell})) = a_k$ 
we have that 
$p((v_\ell, v_{-\ell}), (B_\ell, B_{-\ell})) = p_{a_k} \le B_\ell$. 
\item
For every $v_\ell$ and every $B_\ell$  we have that 
\[ 
f((v_\ell, v_{-\ell}), (B_\ell, B_{-\ell})) \in 
\mbox{\rm argmax}_{ \ a_k \in S \; } 
\{v_\ell(a_k) - p_{a_k}\}, 
\]
where $S = \{ a_k \in  \mathcal{A}(v_{-\ell}, B_{-\ell}) \; | \;   p_{a_k} \le B_{\ell} \}$
\end{enumerate}

\begin{proof}
Let $(v, B)=((v_\ell, \ v_{-\ell}), (B_\ell, \ B_{-\ell}))$ and
$(v', B')=((v'_\ell, \ v_{-\ell}), (B'_\ell, \ B_{-\ell}))$. 

First, suppose that $f(v, B) = f(v', B')$ 
but $p_\ell(v, B) > p_\ell(v', B')$. 
Then, $v_\ell(f(v, B)) \ - \ p_\ell(v, B) < v_\ell(f(v', B)) \ - \ p_\ell(v', B')$,
while by budget feasibility (BF) we have 
$p_\ell(v', B') < p_\ell(v, B) \le B_\ell$,
contradicting incentive compatibility (IC). 

Second, suppose that 
$f(v, B) \notin 
\mbox{\rm argmax}_{ \ a_k \in S\;}
 \{v_\ell(a_k) - p_{a_k}\}$ 
 and let $f(v', B') = a_i$ where $a_i \in \mbox{\rm argmax}_{ \ a_k \in S\;} \{v_\ell(a_k) - p_{a_k}\}$.  
As before,  $v_\ell(f(v, B)) \ - \ p_\ell(v, B) < v_\ell(f(v', B)) \ - \ p_\ell(v', B)$, 
where   $p_\ell(v', B') = p_{a_i}\le B_\ell$, 
contradicting incentive compatibility (IC). \qed 
\end{proof}

\begin{claim}\label{claim-delta-neq-neg-infty}
If $f: \mathcal{V} \times \mathcal{B}\rightarrow \mathcal{A}$ is implementable then $\delta(a, a')\ge  p_{a} - p_{a'} > -\infty$, for all
$a, a' \in \mathcal{A}$. 
\end{claim}

\begin{proof}
If $\delta(a, a') = \infty$, the claim trivially holds. 
Otherwise, for every sufficiently small $\epsilon > 0$ there exist $v$ and $B \ge \beta(a')$, 
where $f(v, B)=a$ and $v(a)-v(a')\le \delta(a, a') + \epsilon$. 

By Lemma~\ref{lemma:taxation-Principle}, no positive transfer (NPT) and~(\ref{eqn:delta-private}), we have $p(v, B) = p_{a}$, where
 $0 \le p_{a} \le \beta(a) \le B$ and  $0 \le p_{a'} \le \beta(a') \le B$. 
 
Now, 
$\delta(a, a') + \epsilon \ge v(a) - v(a') \ge p_{a} - p_{a'}$
by incentive compatibility (IC). Therefore,
$\delta(a, a')  \ge p_{a} - p_{a'} > -\infty$, as required.  \qed 
\end{proof}

\begin{claim}\label{claim-private-no-negative-cycle-through-1-|A|}
Let $C$ be a finite cycle in the graph $G_{f, \mathcal{B}}$. 
If $f: \mathcal{V} \times \mathcal{B}\rightarrow \mathcal{A}$ is implementable and   $l(i, k) = \delta(a_i, a_k)$ 
for every arc $(i, k) \in C$ then $C$ is a nonnegative cycle. 
\end{claim}

\begin{proof}
We prove the claim by contradiction. 
Suppose there exists a finite  negative  cycle $C$ 
in the graph $G_{f, \mathcal{B}}$
such that  $l(i, k) = \delta(a_i, a_k)$ 
for every arc $(i, k) \in C$.
We can assume without loss of generality that $C$ is  
a {\it simple cycle} (with no repeated nodes).
Otherwise, we can split it into two cycles, where one of them must clearly be negative. 
The process repeats until we are left with a negative  simple cycle. 

Without loss of generality (by renaming outcomes if necessary), we assume that
$C = (1, 2, \ldots, k)$, where $k \le |\mathcal{A}|$ . In particular, 
$l(1,2) + l(2, 3) + \cdots + l(k-1, k) + l(k, 1) = 
\delta(a_1, a_2)  + \cdots + \delta(a_{k-1}, a_{k}) + \delta(a_k, a_1) < 0$.

By Claim~\ref{claim-delta-neq-neg-infty}, $\delta(a_i, a_{i+1})  
\ge p_{a_i} - \; p_{a_{i+1}}$. Adding these inequalities together leads to
\[
\delta(a_1,a_2) + \delta(a_2, a_3) + \cdots + \delta(a_{k-1}, a_{k}) + \delta(a_k, a_1)  \ge p_{a_1} - p_{a_2} + p_{a_2} - p_{a_3} + \cdots - p_{a_1} =  0.
\]
 But this contradicts the assumption that $C$ is a 
 negative  cycle.\qed 
\end{proof}

\begin{claim}\label{claim-private-no-negative-cycle-through-zero}
Let $C$ be a finite cycle in the graph $G_{f, \mathcal{B}}$. 
If $f: \mathcal{V} \times \mathcal{B}\rightarrow \mathcal{A}$ is implementable and   $l(i, k) = \theta(a_i)$ 
for some arc $(i, k) \in C$ then $C$  is a nonnegative cycle.
\end{claim}

\begin{proof}
Suppose the claim is false.  
Let $C$ be a negative cycle in the graph $G_{f, \mathcal{B}}$ 
and let $I^+ = \{ (i, k) \in C \; | \; l(i, k) = \theta(a_i)\}$.
By Claim~\ref{claim-private-no-negative-cycle-through-1-|A|}, the set  $I^+$ is nonempty. 

We can assume that $C$ is a simple cycle (by Claim~\ref{claim-private-no-negative-cycle-through-1-|A|} and its proof).  
Without loss of generality (by renaming outcomes if necessary), we can further assume that
$C = (1, 2, \ldots, m)$. Therefore,    
$l(1,2) + l(2, 3) + \cdots + l(m-1, m) + l(m, 1) < 0$.
Now, since $l(e) \ge 0$, for every arc $e \in I^+$, there must exist $i < k < m$
such that  $l(i, i+1) + l(i+1, i+2) + \cdots + l(k, k+1) < 0$, where 
  $(k, k+1) \in I^+$ and $(i, i+1), (i+1, i+2), \ldots, (k-1, k) \notin I^+$.
In particular, 
$l(i,i+1) + \cdots + l(k, k+1) 
= \delta(a_i,a_{i+1}) + \cdots + \delta(a_{k-1}, a_k) + \theta(a_k) < 0$.

By Claim~\ref{claim-delta-neq-neg-infty}, $0 > \delta(a_i, a_{i+1}) + \cdots + \delta(a_{k-1}, a_k)  + \theta(a_k)  \ge p_{a_i} - p_{a_k} + \theta(a_k)$.  This implies that  $p_{a_k}  > p_{a_i}  + \theta(a_k)$.
By no positive transfer (NPT), we have  $p_{a_i} \ge 0$ and 
therefore $p_{a_k}  > \theta(a_k)$. Consider two cases:

{\bf Case 1.}  Assume first that  $\theta(a_k) = \beta(a_k)$.  
By definition, there exist $v, B$ and 
a small enough $\epsilon \ge 0$ 
such that  $f(v, B) = a_k$, where $B = \beta(a_k) + \epsilon$.
Therefore, $p_{a_k} > \beta(a_k) + \epsilon$ , 
 contradicting budget feasibility (BF).

{\bf Case 2.}  Next assume that 
$\theta(a_k) = \omega(a_k)$. By definition, there exist $v, B$ and 
a small enough $\epsilon \ge 0$ 
such that  $f(v, B) = a_k$, where $v(a_k) = \omega(a_k) + \epsilon$.  
Therefore,  $p_{a_k}  > \omega(a_k) + \epsilon$, 
contradicting individual rationality (IR).\qed 
\end{proof}

\subsection{Proof of Theorem~\ref{thm:suff-private-budgets-no-over-report}}

The first direction is by Theorem~\ref{thm:necessary-cond-private-budgets}. The other direction follows from the next two claims. We start by defining one more graph.

\begin{definition}[The Graph $H_{f, \mathcal{B}}$]\label{def:Graph-H-private-no-over-report}
Let $M(H_{f, \mathcal{B}}) = \{0, 1, 2, \ldots , |\mathcal{A}|\}$ be 
the node set  of the graph, where 0 is a special node,
and  $\{ 1, 2, \ldots, |\mathcal{A}| \}$ correspond to all possible outcomes. 
Let ${E}(H_{f, \mathcal{B}}) = E_0 \cup E_1 \cup E_2$ 
be the directed arc set of the graph,  where  

\[
E_0 = \{ (i, k) \ | \  i, k > 0 \mbox{ and } \delta(a_i, a_k) < \infty  \},
\]
\[
E_1 = \{ (0, k) \ | \  k > 0 \}, 
\]
\[
E_2 = \{ (i, 0) \ | \  i >  0 \}.
\]

Finally, the {\it length} of a directed arc 
$(i, \; k) \in {E}(H_{f, \mathcal{B}})$ 
is defined as follows: 

\[ l_H(i, k) = \left\{ \begin{array}{ll}
        \delta(a_i, a_k) & \mbox{if $(i, k) \in E_0$}\\
        0 & \mbox{if $(i, k) \in E_1$} \\
        \theta(a_i) & \mbox{if $(i, k) \in E_2$}.\end{array} \right. \] 
\end{definition}

\begin{claim}\label{claim:no-neg-cycle-in-G-implies-no-neg-in-H}
If  $G_{f, \mathcal{B}}$ contains no finite negative cycles then
$H_{f, \mathcal{B}}$ contains no finite negative  cycles.
\end{claim}

\begin{proof}
Suppose by contradiction that there exists 
a finite negative  cycle $C$ in $H_{f, \mathcal{B}}$. Since all arcs in 
${E}(H_{f, \mathcal{B}})$ have length $< \infty$,  
it suffices to show that there exists a finite cycle $C'$  
in $G_{f, \mathcal{B}}$ with a smaller length. 
First assume that $e \in E_0 \cap C$. 
Add $e$ to $C'$, and notice that 
$l(e) \le l_H(e)$. 
Next assume that $e = (i, 0) \in E_2 \cap C$ for some $i > 0$.
Clearly, the  consecutive arc in $C$ must have the form $e' = (0,k) \in E_1$ 
for some $k > 0$.
Add $(i, k)$ to $C'$, and notice that $l(i, k) \le l_H(i, 0) + l_H(0,k)$.\qed 
\end{proof}

\begin{claim}\label{claim-private-B-no-cycles-in-H-graph-implies-truth}
If $H_{f, \mathcal{B}}$ contains no finite negative  cycles
then $f: \mathcal{V} \times \mathcal{B} \rightarrow \mathcal{A}$ for private 
budgets with no budget over-reporting is implementable.
\end{claim}

\begin{proof}
Consider the payment $p_{a_i} = \Delta_{H_{f, B}}(i,0)$, where 
$\Delta_{H_{f, B}}(i,0)$ denotes the length of the shortest
path from node $i > 0$ to node $0$ in the graph $H_{f, B}$.

Since $H_{f, \mathcal{B}}$ contains no negative  cycles, 
we have that $\Delta_{H_{f, B}}(i,0) > -\infty$. 
In addition,  $\Delta_{H_{f, B}}(i,0) \le l_H(i,0) = \theta(a_i) < \infty$, 
since $a_i$ can be obtained for some report of the player (recall our assumption that $f$ is onto).  This shows that  
the payment satisfies individual rationality (IR) and budget feasibility (BF). 
To show no positive transfer (NPT), recall that $l_H(0,i)=0$ 
and so $p_{a_i} = \Delta_{H_{f, B}}(i,0) =  l_H(0,i) +  \Delta_{H_{f, B}}(i,0) \ge  0$, since all cycles have nonnegative length.

To show incentive compatibility (IC), assume for the purpose of contradiction that the player can 
benefit from reporting $v', B'$ instead of its true value $v$ and true budget $B$. 
Specifically, $f(v, B) = a_i$ and  $f(v', B') = a_k$,  
but  $v(a_i) - \Delta_{H_{f, B}}(i,0) <  v(a_k) - \Delta_{H_{f, B}}(k,0)$. 
By budget feasibility at $(v, B)$ and $(v', B')$, 
we have  $\Delta_{H_{f, B}}(i,0) \le \beta(a_i) \le B$ 
and $\Delta_{H_{f, B}}(k,0) \le \beta(a_k) \le B'$.
Importantly, since the player cannot over-report its private budget we have  $\beta(a_k)\le B' \le B$, and therefore
$\delta(a_i, a_k)  \le v(a_i) - v(a_k)
< \infty$.

Now,  
$l_H(i,k) + \Delta_{H_{f, B}}(k,0) = 
\delta(a_i, a_k) + \Delta_{H_{f, B}}(k,0) \le  v(a_i) - v(a_k) + \Delta_{H_{f, B}}(k,0) < \Delta_{H_{f, B}}(i,0)$.
Notice that  $l_H(i,k) + \Delta_{H_{f, B}}(k,0)$ represents a length of a  path 
 from node $i$ to node $0$ through node $k$ in the graph $H_{f, B}$. 
By the above,  $l_H(i,k) + \Delta_{H_{f, B}}(k,0)$ is strictly
smaller than $\Delta_{H_{f, B}}(i,0)$, the length of the shortest path from node $i$ to $0$,
a contradiction to the minimality of $\Delta_{H_{f, B}}(i,0)$.\qed 
\end{proof}

\subsection{Proof of Proposition~\ref{prop:suff-private-budget}}

We start by defining a related graph: 

\begin{definition}[The graph $\widehat{H}_{f, \; \mathcal{B}}$]
Let the {\rm node set} of the graph be 
$M(\widehat{H}_{f, \; \mathcal{B}}) = \{0, 1, 2, \ldots, |\mathcal{A}|\}$. 
Let the {\rm directed arc set} of the graph be
 ${E}(\widehat{H}_{f, \; \mathcal{B}}) = E_0 \cup E_1 \cup E_2$, where 

\[
E_{0} = \{ (i, \; k) \ | \  i, \; k > 0 \mbox{ and }  \widehat{\delta}(a_i, \; a_k) \le \theta(a_i) \},
\]
\[
E_{1} = \{ (0, \; k) \ | \  k > 0 \},
\]
\[
E_{2} = \{ (i, \; 0) \ | \  i > 0 \}.
\]

Finally, the {\it length} of the directed arc $(i,k)$ is defined as follows: 

\[ l(i, k) = \left\{ \begin{array}{ll}
        \widehat{\delta}(a_i, \; a_k) 
        & \mbox{if $(i, \;  k) \in E_{0}$}\\
        0 & \mbox{if $(i, \; k) \in E_{1}$}\\
        \theta(a_i) & \mbox{if $(i, \; k) \in E_{2}$}.\end{array} \right. \] 
\end{definition}

\begin{claim}
If  $\widehat{G}_{f, \mathcal{B}}$ contains no finite negative  cycles then
$\widehat{H}_{f, \mathcal{B}}$ contains no finite negative cycles.
\end{claim}

\begin{proof}
The proof is similar to the proof of Claim~\ref{claim:no-neg-cycle-in-G-implies-no-neg-in-H}.\qed 
\end{proof}

\begin{claim}\label{claim-private-B-no-cycles-in-WIDE_HAT-H-graph-implies-truth}
If $\widehat{H}_{f, \mathcal{B}}$ contains no finite negative  cycles
then $f: \mathcal{V} \times \mathcal{B} \rightarrow \mathcal{A}$ for private budgets  is implementable.
\end{claim}

\begin{proof}
Consider the payment $p_{a_i} = \Delta_{\widehat{H}_{f, \; \mathcal{B}}}(i,0)$,  where 
$\Delta_{\widehat{H}_{f, \; \mathcal{B}}}(i,0)$ denotes the length of the shortest 
path from node $i > 0$ to node $0$ 
in the directed graph $\widehat{H}_{f, \; \mathcal{B}}$. 
 
The first part of the proof is similar to the proof of 
Claim~\ref{claim-private-B-no-cycles-in-H-graph-implies-truth}.
In particular, it shows that $\Delta_{\widehat{H}_{f, \; \mathcal{B}}}(i,0)$  satisfies BA, IR and NPT.

It remains to show the incentive compatibility (IC). 
 Suppose to the contrary that   there exist $v, v', B, B'$ such that 
$f(v, B) = a_i,\; f(v', B') = a_k$ and 
$v(a_i) - \Delta_{\widehat{H}_{f, \; \mathcal{B}}}(i, 0) < v(a_k) - \Delta_{\widehat{H}_{f, \; \mathcal{B}}}(k, 0)$,
 where $\Delta_{\widehat{H}_{f, \; \mathcal{B}}}(k, 0) \le B$.

Now, if $(i,k) \in  E(\widehat{H}_{f, \; \mathcal{B}})$, 
then by rearranging we have that 
\[
l(i, k) + \Delta_{\widehat{H}_{f, \; \mathcal{B}}}(k, 0) = 
\widehat{\delta}(a_i, \; a_k) + \Delta_{\widehat{H}_{f, \; \mathcal{B}}}(k, 0) \le
v(a_i) - v(a_k) + \Delta(k,0) < \Delta_{\widehat{H}_{f, \; \mathcal{B}}}(i,0).
\]  

The left-hand side represents a length of a direct path 
from $i$ to $0$ (through $k$) which is 
strictly smaller than $\Delta_{\widehat{H}_{f, \; \mathcal{B}}}(i,0)$, 
the length of the shortest path from $i$ to $0$,
a contradiction to the minimality of $\Delta_{\widehat{H}_{f, \; \mathcal{B}}}(i,0)$.

Otherwise, $(i,k) \notin E(\widehat{H}_{f, \; \mathcal{B}})$, 
and thus 
$\widehat{\delta}(a_i, \; a_k) > \theta(a_i)$. By BA and IR we have  

\[
v_i(a_i) - v_i(a_k) + \Delta_{\widehat{H}_{f, \; \mathcal{B}}}(k,0) < \Delta_{\widehat{H}_{f, \; \mathcal{B}}}(i,0)
\le \theta(a_i) < 
\widehat{\delta}(a_i, \; a_k) \le v_i(a_i) - v_i(a_k).
\]

By rearranging we have that  $\Delta_{\widehat{H}_{f, \; \mathcal{B}}}(k,0) < 0$, a contradiction to NPT.\qed 
\end{proof}

\subsection{Proof of Theorem~\ref{thm:rev-equiv}}

The proof is the consequence of the following claims.  

\begin{claim}\label{claim-rev-equiv-necessary}
Let $f: \mathcal{V} \times \mathcal{B}\rightarrow \mathcal{A}$ 
be a generically implementable social choice function for private budgets.
If $f$ satisfies the revenue equivalence principle 
then  $\Delta_{G_{f, \; \mathcal{B}}}(i, k)= -\Delta_{G_{f, \; \mathcal{B}}}(k, i)$
 for all $a_i, a_k \in \mathcal{A}$ with $\beta(a_i)=\beta(a_k)$.
\end{claim}

\begin{proof}
By Theorem~\ref{thm:necessary-cond-private-budgets},  the graph $G_{f, \; \mathcal{B}}$ contains no negative length cycles, and therefore $\Delta_{G_{f, \; \mathcal{B}}}(i, k)+ \Delta_{G_{f, \; \mathcal{B}}}(k, i) \ge 0$ for all $a_i, a_k \in \mathcal{A}$. If $\beta(a_i)=\beta(a_k)=0$, then $\theta(a_i)=\theta(a_k)=0$ (since the valuation is nonnegative) and therefore $l(i, k) \le  0$, and $l(k, i) \le 0$. Observe that none of these arcs can be negative 
(by Theorem~\ref{thm:necessary-cond-private-budgets}), and therefore $\Delta_{G_{f, \; \mathcal{B}}}(i, k) + 
 \Delta_{G_{f, \; \mathcal{B}}}(k, i) = l(i,k) +l(k, i) = 0$, as required.

Suppose, by way of contradiction, that  $\beta(a_i)=\beta(a_k) > 0$ and 
 $\Delta_{G_{f, \; \mathcal{B}}}(i, k) + 
 \Delta_{G_{f, \; \mathcal{B}}}(k, i) > 0$
 for some $a_i, a_k \in \mathcal{A}$. 
  Let $p$  be a payment function that generically implements $f$ at $a_i$. 
 We  now construct a corresponding directed graph $\overrightarrow{G}$ defined as follows. The set of nodes is $\{ j \; | \; \beta(a_j)=\beta(a_i), \; a_j \in \mathcal{A} \}$. For every pair of nodes $j, j'$, there is an arc $(j, j')$ in $\overrightarrow{G}$  if and only if $p(a_j) - p(a_{j'}) = \delta(a_{j}, a_{j'})$.  We next look at the strongly connected components of $\overrightarrow{G}$. 
 
 We first claim that $i$ and $k$ belong to distinct strongly connected components of  $\overrightarrow{G}$. Suppose not. There is thus
  a path $P$ from $i$ to $k$, and a path $Q$ from $k$ to $i$ in
   $\overrightarrow{G}$. Clearly, $P$ and $Q$ are also paths in 
   $G_{f, \; \mathcal{B}}$ (not necessarily the  shortest paths), and thus $\Delta_{G_{f, \; \mathcal{B}}}(i, k) \le l(P)$, and  $\Delta_{G_{f, \; \mathcal{B}}}(k, i) \le l(Q)$ (where $l(P), l(Q)$ denote the length of the paths $P, Q$ in the graph ${G_{f, \; \mathcal{B}}}$, respectively). In particular, $0 < l(P) + l(Q)$. 
   
   Now, for every arc $(j, j') \in \overrightarrow{G}$ 
   we have  
   $p(a_j) - p(a_{j'}) = \delta(a_{j}, a_{j'})$.
   Summing these equalities along the paths gives us 
   $0 = \Sigma_{\; (j, \; j') \in P \;} \delta(a_{j}, a_{j'}) + \Sigma_{\; (j, \; j') \in Q \;} \delta(a_{j}, a_{j'})$.  
   But since $\delta(a_{j}, a_{j'}) \ge \min \{ \delta(a_{j}, a_{j'}), \theta(a_j) \} = l(j, j')$, 
   we obtain $0 = \Sigma_{\; (j, \; j') \in P \;} \delta(a_{j}, a_{j'}) + \Sigma_{\; (j, \; j') \in Q \;} \delta(a_{j}, a_{j'}) \ge l(P) + l(Q)$, a contradiction. We conclude that the graph $\overrightarrow{G}$ has at least 2 strongly connected components.
   
Let $K$ be a strongly connected component of  $\overrightarrow{G}$ with no  arcs outgoing from a node in $K$ to a node in some other strongly connected component $K'$ of  $\overrightarrow{G}$. 
It is easy to check that every directed graph has at least one such component (e.g., since the  component graph $\overrightarrow{G}^{SCC}$ of 
$\overrightarrow{G}$ is a directed acyclic graph~\cite[Chapter 22]{CLR-Book}), so that $K$ is well defined. Based on $K$ and $p$ we define the following payment function: 

\[ p'(a)  = \left\{ \begin{array}{ll}
         p(a)+\epsilon  & \mbox {if } a \in K\\
        p(a) & \mbox{otherwise}. \\
        \end{array} \right. \]

Recall that $p$ generically 
implements $f$ at $a_i$ and so   
the payment function $p'$ satisfies IR, BF and NPT, for some small enough $\epsilon > 0$ (by Definition~\ref{def:Generically-implementable}, part~(2)).

We now show that $p'$ satisfies incentive compatibility (IC).
Let $f(v, B)=a$ and $f(v', B')=a'$. 
It suffices to show that 
$v(a) - p'(a) \ge v(a') - p'(a')$ for  every $a \in K$ and every 
$a' \in \mathcal{A}$ with $p'(a') \le B$.
Equivalently, we need to show that 
$v(a) - v(a') \ge p'(a)   - p'(a')$ for  every $a \in K$ and every 
$a' \in \mathcal{A}$ with $p'(a') \le B$.\footnote{Clearly if $a' \notin K$, $p(a') = p'(a')$ and therefore $p(a') \le B$ if and only if  $p'(a') \le B$. 
Otherwise, $p(a') + \epsilon = p'(a')$,
so that for a small enough $\epsilon > 0$,
$p'(a') < \theta(a') \le \beta(a')=\beta(a)\le B$ (by Definition~\ref{def:Generically-implementable}, part~(2)), thus 
$p(a') \le  B$
if and only if  
$p'(a') \le B$.}

We begin with the case in which $\beta(a') > \beta(a)$. 
Since $p$ generically implements $f$ at $a_i$, then 
for some small enough $\epsilon > 0$, we have that 
$p'(a)- p'(a') < v(a) - v(a')$
for every $a'$ with $\beta(a') > \beta(a)$ (by Definition~\ref{def:Generically-implementable}, part~(4)).

Otherwise, $\beta(a') \le \beta(a)$. By~(\ref{eqn:delta-private}), 
 $\delta(a, a') \le v(a) - v(a')$, and thus it suffices to show in all the remaining cases that $p'(a) - p'(a') \le \delta(a, a')$ for every $a' \in \mathcal{A}$.
 
Now, if $\beta(a') <  \beta(a)=\beta(a_i)$, then 
by Claim~\ref{claim-delta-neq-neg-infty} and 
Definition~\ref{def:Generically-implementable}, part~(3), we have that $p(a) - p(a') < \delta(a, a')$. Therefore, for some small enough $\epsilon > 0$, $p'(a) - p'(a') < \delta(a, a')$, as required.  

For the remaining case $\beta(a') = \beta(a)$, there are two subcases to consider, according to whether $a' \in K$ or not. 
In the first subcase
$p'(a) = p(a) + \epsilon$ and $p'(a') = p(a')+\epsilon$. 
By Claim~\ref{claim-delta-neq-neg-infty}, 
 $p(a) - p(a') \le \delta(a, a')$, and therefore  
$p'(a) - p'(a') = p(a) - p(a') \le \delta(a, a')$,
as required.   
 
In the other subcase, $a' \notin K$.
By Claim~\ref{claim-delta-neq-neg-infty} and the fact that $K$ has no outgoing arc in $\overrightarrow{G}$, we have  
$p(a) - p(a') < \delta(a, a')$,  and thus for a small enough $\epsilon > 0$, $p'(a) - p'(a') < \delta(a, a')$, as required. 

Finally, since $\overrightarrow{G}$ has at least 2 strongly connected components there exist $a \in K$ and $a'\notin K$ such that $\beta(a) =\beta(a')$. However, $\epsilon = p'(a) - p(a) \neq p'(a') - p(a') = 0$, contradicting the assumption that the revenue equivalence principle is satisfied.\qed 
\end{proof}

\begin{claim}\label{claim-rev-equiv-sufficieny}
Let $f: \mathcal{V} \times \mathcal{B}\rightarrow \mathcal{A}$ be an  implementable social choice function for private budgets with no budget over-reporting.
If   $\Delta_{G_{f, \; \mathcal{B}}}(i, k)= -\Delta_{G_{f, \; \mathcal{B}}}(k, i)$
 for all $a_i, a_k \in \mathcal{A}$ with $\beta(a_i)=\beta(a_k)$ then 
$f: \mathcal{V} \times \mathcal{B}\rightarrow \mathcal{A}$ 
satisfies the revenue equivalence principle.
\end{claim}

\begin{proof}
Let  $p, p'$ be payments such that the mechanisms $(f, p)$ 
and $(f, p')$ are implementable with private budgets. 
Suppose $\Delta_{G_{f, \; \mathcal{B}}}(i, k)= 
-\Delta_{G_{f, \; \mathcal{B}}}(k, i)$
 for all $a_i, a_k \in \mathcal{A}$ with $\beta(a_i)=\beta(a_k)$. 
 By Theorem~\ref{thm:necessary-cond-private-budgets} the graph $G_{f, \; \mathcal{B}}$ contains no negative length cycles, and therefore all shortest paths in this graph are finite.
 Let $P$ be a shortest path from node $i$ to node $k$ in 
 ${G_{f, \; \mathcal{B}}}$, and let $e=(j, j+1)$ be 
 some arc in $P$.
 
 By Claim~\ref{claim-delta-neq-neg-infty}, 
 $p_{a_{j}} - p_{a_{j+1}} \le \delta(a_{j}, a_{j+1})$.  
 Now, $p_{a_{j}} \le \theta(a_{j})$ (by IR and BF) and 
 $p_{a_{j+1}} \ge 0$ (by NPT) and so $p_{a_{j}} - p_{a_{j+1}} \le \theta(a_{j})$, as well. 
 Thus  $p_{a_{j}} - p_{a_{j+1}} \le \min\{\delta(a_{j}, a_{j+1}),  \theta(a_{j})\} = l(j, j+1)$.  Summing over all arcs in $P$ we have  $p_{a_{i}} - p_{a_{k}} \le l(P)=\Delta_{G_{f, \; \mathcal{B}}}(i, k)$.  Similarly, $p_{a_{k}} - p_{a_{i}} \le \Delta_{G_{f, \; \mathcal{B}}}(k, i)$. Therefore, $-\Delta_{G_{f, \; \mathcal{B}}} (k, i) \le 
 p_{a_{i}} - p_{a_{k}} \le \Delta_{G_{f, \; \mathcal{B}}}(i, k)$. 
 Since, $-\Delta_{G_{f, \; \mathcal{B}}} (k, i) = \Delta_{G_{f, \; \mathcal{B}}}(i, k)$  we have that 
 $p_{a_{i}} - p_{a_{k}} = \Delta_{G_{f, \; \mathcal{B}}}(i, k)$. By the same argument we have that $p'_{a_{i}} - p'_{a_{k}} = \Delta_{G_{f, \; \mathcal{B}}}(i, k)$. Therefore,  $p_{a_{i}} - p_{a_{k}} = p'_{a_{i}} - p'_{a_{k}}$ for $a_i, a_k \in \mathcal{A}$  with $\beta(a_i)=\beta(a_k)$, and thus $f$ satisfies the revenue equivalence principle.\qed 
\end{proof}

\subsection{Proofs for Section~\ref{Sec:example-and-application}}

\noindent{\bf Proof of Claim~\ref{claim-VCG-possibility-without-money}: }
Without loss of generality we can assume that 
there are only two possible outcomes: $\mathcal{A} = \{a, b\}$. 
Let $w(a) = (\sum_{\ell \neq \ell'} \kappa_\ell \cdot v_\ell(a)) + \gamma_a$
and $w(b) = (\sum_{\ell \neq \ell'} \kappa_\ell \cdot v_\ell(b)) + \gamma_b$. 
 
If  $w(a) = w(b)$, then  
$\Delta_{\ell'}(a, b, v_{-\ell'})= 0, \;  \Delta_{\ell'}(b, a, v_{-\ell'})=1$. 

If $w(a) + 1  = w(b)$, then 
$\Delta_{\ell'}(a, b, v_{-\ell'})= 1, \;  \Delta_{\ell'}(b, a, v_{-\ell'})=0$.

Finally, if $w(a) \ge  w(b) + 1$, then  
$\Delta_{\ell'}(a, b, v_{-\ell'})= -\infty$. Note that this is not a violation of the  
the condition in Proposition~\ref{thm-Characterization-for-zero-budgets} 
since $f$ is not onto in this case. 
 The case where $w(a) + 1 < w(b)$ is similar.\qed 

\vspace{0.4cm}

\noindent{\bf Proof of Claim~\ref{prop-VCG-impossibility-public}: }
Suppose not. Fix $f(v_1, \ldots, v_n) \in \mbox{\rm argmax}_{ \ a' \in \mathcal{A}\;} 
\{ \sum_\ell \kappa_\ell \cdot v_\ell(a') + \gamma_{a'}\}$, 
such that  $\kappa_1 > 0, \ldots, \kappa_n > 0$ and  
$\gamma_{a'} \in \bbbr$ in an arbitrary manner. 
Without loss of generality (by renaming players and outcomes if necessary), we assume that
$B_1 < \infty$ and  that $\gamma_a \ge \gamma_b $, where $a, b \in \mathcal{A}$.  
We first prove for the case of two players. 
Consider the following valuations:

\[ v_1(x)  = \left\{ \begin{array}{ll}
         B_1+1 & \mbox {if } x=a\\
        0 & \mbox{otherwise}. \\
        \end{array} \right. \] 
Additionally,  
  
\[ v_2(x)  = \left\{ \begin{array}{ll}
        \kappa_1 / \kappa_2\cdot(B_1 + 1 + L) & \mbox {if }  x=a\\
        \kappa_1 / \kappa_2 \cdot (2(B_1+1) + L ) + (\gamma_a - \gamma_b)/ \kappa_2  & \mbox {if } x=b \\ 
        0 & \mbox{otherwise}. \\
        \end{array} \right. \]

Note that for sufficiently large  $L > 0$ and every $c \in \mathcal{A}$ we have
\[       
\gamma_a + \sum_i \kappa_i v_i(a) \; = \; \gamma_b + \sum_i \kappa_i v_i(b)  \; \ge \;
\gamma_c + \sum_i \kappa_i v_i(c).
\] 
  
Now, fix $v_1, v_2$. 
If player 1 increases $v_1(b)$ by $ \epsilon > 0 $ then $b$ will be chosen, and 
therefore: $\delta_1(b, a) \le v_1(b) - v_1(a) = -(B_1 +1)$. 
Additionally, if player 1 increases $v_1(a)$ by $ \epsilon > 0 $ then $a$ will be chosen
and therefore $\theta_1(a)  = \min \ \{ B_1, \  \inf \ \{v'_1(a) \ | \ v'_1 \in \mathcal{V}_1 \mbox{ such that } 
f(v'_1, v_2) = a\} \} = \min \ \{ B_1, \ B_1 +1 \} = B_1$.

Now, $0 > \delta_1(b, a) + \theta_1(a) \ge l(b,a) + l(a, b)$
and therefore the graph $G_{f, B}$ has a negative cycle. 
By Theorem~\ref{thm:suff-private-budgets-no-over-report}, 
  $f$ is not implementable, a contradiction. 
To prove the theorem for $n > 2$ players , we can add $n-2$ players 
with zero valuations and arbitrary public budgets.\qed

\vspace{0.4cm}

\end{document}